\documentclass[a4paper,UKenglish,autoref]{lipics-v2019}

\usepackage[ruled,vlined,linesnumbered,commentsnumbered]{algorithm2e}
\usepackage{xcolor}
\usepackage{hyperref}

\title{A logarithmic approximation algorithm for the activation edge multicover problem}

\titlerunning{A logarithmic approximation algorithm for the activation edge multicover problem}

\author{Zeev Nutov}{The Open University of Israel}{nutov@openu.ac.il}
{https://orcid.org/0000-0002-6629-3243}{}
\authorrunning{Zeev Nutov}
\Copyright{Zeev Nutov}

\nolinenumbers

\ccsdesc[100]{Theory of computation~Design and analysis of algorithms}

\acknowledgements{}

\begin{document}

\maketitle
\newcommand {\ignore} [1] {}

\def\sem    {\setminus}
\def\subs   {\subseteq}
\def\empt  {\emptyset}

\def\di {\displaystyle}

\newcommand{\f}   {\frac}

\def\t {\tilde}
\def\h {\hat}

\def\A {\mathbb{A}}

\def\al   {\alpha}
\def\be   {\beta}
\def\de   {\delta}
\def\eps {\epsilon}
\def\ga {\gamma}
\def\th {\theta}

\def\AA  {{\cal A}}
\def\BB  {{\cal B}}
\def\CC  {{\cal C}}
\def\FF  {{\cal F}}
\def\SS  {{\cal S}}


\def\aec    {{\sc Activation Edge-Cover}}
\def\aem   {{\sc Activation Edge-Multicover}}
\def\baem {{\sc Bipartite Activation Edge-Multicover}}

\def\kcs {{\sc $k$-Connected Subgraph}}
\def\akcs {{\sc Activation $k$-Connected Subgraph}}

\def\kos {{\sc $k$-Out-Connected Subgraph}}
\def\akos {{\sc Activation $k$-Out-Connected Subgraph}}

\keywords{edge multicover, activation problem, minimum power}

\begin{abstract}
In the {\sc Activation Edge-Multicover} problem  we are given a multigraph $G=(V,E)$ 
with activation costs $\{c_{e}^u,c_{e}^v\}$ for every edge $e=uv \in E$, 
and degree requirements $r=\{r_v:v \in V\}$.
The goal is to find an edge subset  $J \subs E$ of minimum activation cost 
$\sum_{v \in V}\max\{c_{uv}^v:uv \in J\}$,
such that every $v \in V$ has at least $r_v$ neighbors in the graph $(V,J)$.
Let $k= \max_{v \in V} r_v$ be the maximum requirement
and let $\displaystyle \theta=\max_{e=uv \in E} \f{\max\{c_e^u,c_e^v\}}{\min\{c_e^u,c_e^v\}}$ 
be the maximum quotient between the two costs of an edge.
For $\th=1$ the problem admits approximation ratio $O(\log k)$.
For $k=1$ it generalizes the {\sc Set Cover} problem (when $\theta=\infty$), 
and admits a tight approximation ratio $O(\log n)$. 
This implies approximation ratio $O(k \log n)$ for general $k$ and $\theta$,
and no better approximation ratio was known.  
We obtain the first logarithmic approximation ratio $O(\log k +\log\min\{\theta,n\})$,
that bridges between the two known ratios -- 
$O(\log k)$ for $\th=1$ and $O(\log n)$ for $k=1$. 
This implies approximation ratio 
$O\left(\log k +\log\min\{\theta,n\}\right) +\beta \cdot (\theta+1)$ 
for the {\sc Activation $k$-Connected Subgraph} problem,
where $\beta$ is the best known approximation ratio 
for the ordinary min-cost version of the problem. 
\end{abstract}

\section{Introduction} \label{s:intro}

In network design problems one seeks a cheap subgraph that satisfies a prescribed property.
A traditional setting, motivated by wired networks, is when each edge or node has a cost, 
and we want to minimize the cost of the subgraph.  
In wireless networks a communication between two nodes depends on our ''investment'' in both nodes -- 
like transmission energy and different types of equipment.
The node weighted setting captures just some of these scenarios. 
In 2011 Panigrahi \cite{P} suggested a generalization, that captures many possible wireless networks scenarios.
In Panigrahi's model, every edge $uv$ 
has an activating function $f(x_u,x_v)$ to $\{0,1\}$, such that an edge $uv$ is activated 
if and only we invest $x_u$ at node $u$ and $x_v$ at node $v$ such that $f(x_u,x_v)=1$. 
Here we use a simpler but less general setting suggested in \cite{KNS}, 
which is equivalent to that of Panigrahi for problems 
in which inclusion minimal feasible solutions have no parallel edges
(but the input graph may have parallel edges).

More formally, in {\bf activation network design problems} we are given 
an undirected (multi-)graph $G=(V,E)$ where every edge $e=uv \in E$ has two (non-negative) {\bf activation costs} $\{c_e^u,c_e^v\}$;
here $e=uv \in E$ means that the edge $e$ has ends $u,v$ and belongs to $E$. 
An edge $e=uv \in E$ is {\bf activated by a level assignment} $\{l_v: v \in V\}$ to the nodes if $l_u \geq c_e^u$ and $l_v \geq c_e^v$. 
The goal is to find a level assignment of minimum value $l(V)=\sum_{v \in V} l_v$, such that 
the activated edge set $J=\{e=uv \in E:c_e^u \leq l_u, c_e^v \leq l_v\}$ satisfies a prescribed property.
Equivalently, the minimum value level assignment that activates an edge set $J \subs E$ is given by $\ell_J(v)=\max \{c_e^v: e \in \delta_J(v)\}$;
here $\delta_J(v)$ denotes the set of edges in $J$ incident to $v$, and a maximum taken over an empty set is assumed to be zero.
We seek an edge set $J \subs E$ that satisfies the given property and minimizes $\ell_J(V)=\sum_{v \in V} \ell_J(v)$.
Note that while we use $l_v$ to denote a level assignment to a node $v$, 
we use a slightly different notation $\ell_J(v)$ for the function that evaluates the optimal 
assignment that activates a given edge set $J$. 

Let $G=(V,E)$ be a multigraph. 
Given degree requirements $r=\{r_v:v \in V\}$ we say that an edge set $J$  is an {\bf $r$-edge-cover} if 
every $v \in V$ has at least $r_v$ neighbors in the graph $(V,J)$. 
We consider the following problem.

\begin{center} \fbox{\begin{minipage}{0.98\textwidth} \noindent
\underline{\aem} \\
{\em Input:} \ \ A multigraph $G=(V,E)$ with activation costs $\{c_e^u,c_e^v\}$ for every edge $e=uv \in E$, 
\hphantom{Input:} \ \ and degree requirements $r=\{r_v:v \in V\}$. \\
{\em Output:}   An $r$-edge-cover  $J \subs E$ of minimal activation cost
$\displaystyle \ell_J(V)=\sum_{v \in V}\max_{uv \in E}c_{uv}^v$.
\end{minipage}}\end{center}

Equivalently, {\aem} can be cast as a problem of assigning {\bf levels} $\{l_v: v \in V\}$ to the nodes 
of minimum total value $l(V)=\sum_{v \in V} l_v$, such that 
the edge set $J=\{uv \in E:c_{uv}^u \leq l_u, c_{uv}^v \leq l_v\}$ activated by the assignment is an $r$-edge-cover.
The {\bf slope} $\th$ of an instance of an activation network design problem is the maximum ratio between the two costs of an edge, namely
\[
\displaystyle \th=\max_{e=uv \in E} \f{\max\{c_e^u,c_e^v\}}{\min\{c_e^u,c_e^v\}} \ .
\]
Two main types of activation costs were extensively studied in the literature. 
\begin{itemize}
\item
{\bf Node weights}.  For all $v \in V$, $c_e^v$ are identical for all edges $e$ incident to $v$. 
This is equivalent to having node weights $w_v$ for all $v \in V$ with the goal of finding a node subset
$V' \subs V$ of minimum total weight $w(V') = \sum_{v \in V'}w_v$ such that the subgraph induced by
$V'$ satisfies the given property.
Note that we may have $\th=\infty$ in this case.
\item
{\bf Power costs}: For all $e = uv \in E$, $c_e^u=c_e^v$.
This is equivalent to having ``power costs'' $c_e = c_e^u=c_e^v$ for all $e = uv \in E$. 
The goal is to find an edge subset $J \subs E$ of minimum total power 
$\sum_{v \in V} \max\{c_e:e \in \de_J(v)\}$ that satisfies the given property.
Note that this the case $\th=1$. 
\end{itemize}

Node weighted problems include many fundamental problems such as 
{\sc Set Cover}, {\sc Node-Weighted Steiner Tree}, and {\sc Connected Dominating Set} c.f. \cite{S,KR,GH}.
Min-power problems were studied already in the 90's, c.f. \cite{SRS,WNE,RM,KKKP}, followed by many more.
They were also widely studied in directed graphs, 
usually under the assumption that to activate an edge one needs to assign power only to its tail, 
while heads are assigned power zero, c.f. \cite{KKKP,N-powcov,HKMN,N-sur}. 
The undirected case has an additional requirement - we want the network to be bidirected,
to allow a bidirectional communication.

Let $k= \max_{v \in V} r_v$ denote the maximum requirement.
Kortsarz, Mirrokni, Nutov, and Tsanko \cite{KMNT} gave 
an $O(\log n)$-approximation algorithm for the min-power version, and 
Cohen \& Nutov \cite{CN} improved the approximation ratio to $O(\log k)$.
However, the node-weighted version is {\sc Set Cover} hard even for $k=1$, 
and thus has approximation threshold $\Omega(\log n)$. 
{\aem} admits an easy approximation ratio $O(\log n)$ for $k=1$ 
and this implies ratio $O(k\log n)$ for any $k$. 
Note that this approximation ratio is not even poly-logarithmic, if $k$ is large.  
We obtain the first (poly-)logarithmic approximation ratio that generalizes 
the known $O(\log k)$-approximation of \cite{CN,KMNT} for the case $\th=1$.

\begin{theorem} \label{t}
The {\aem} problem admits approximation ratio $O\left(\log k +\log\min\{\th,n\}\right)$.
\end{theorem}

The proof of Theorem~\ref{t} has two main ingredients. 
\begin{itemize}
\item
We will show that it is sufficient to prove approximation ratio $O(\log k+\log \th)$, since it implies approximation ratio $O(\log n)$.
In fact, we will give a generic reduction that applies on any activation network design problem: 
achieving approximation ratio $\rho$ on instances with $\th \leq n(\rho+1)/\eps$
implies approximation ratio $\rho+2\eps$ for general instances. 
\item
We will show that the problem admits approximation ratio $O(\log k+\log \th)$. 
The algorithm extends the $O(\log k)$-approximation algorithm of \cite{KMNT,CN}
for the case $\th=1$ to the case of arbitrary $\th$. 
\end{itemize}

Our results and techniques show the advantages of studying the approximability 
of activation network design problem being parameterized by the slope, and 
suggest that some network design activation problems may be easier than they seem. 
We will illustrate this by two examples, as follows. 

A graph is {\bf $k$-outconnected from $s$} if it contains $k$ internally disjoin $sv$ paths for every node $v$.
A graph is {\bf $k$-connected} if it has at least $k+1$ nodes and contains $k$ internally-disjoint paths between every two nodes. 
In the {\kcs} problem we are given a graph $G$ with edge costs and an integer $k$, 
and seek a minimum cost $k$-connected spanning subgraph $H$ of $G$. 
In the activation version {\akcs} problem, instead of ordinary edge costs we have activation costs 
and $H$ should have minimum activation cost.
The {\kos} and {\akos} problems are defined similarly, where $H$ should be $k$-out-connected. 
It is known that if $\deg_J(v) \geq k-1$ for every node $v$, and if $F$ is an inclusion minimal 
edge set such that $J \cup F$ is $k$-connected, or is $k$-out-connected from some node, then $F$ is a forest;
for $k$-connected graph this follows from Mader's Critical Cycle Theorem \cite{Mad-und}, 
while an analogous result was proved for $k$-out-connected graphs in \cite{CJN}. 
It is known that if $F$ is a forest then $\sum_{e \in E} \ell_e(V)/\ell_F(V) \leq \th+1$ \cite[Lemma 15.1(ii) and Corollary 15.1(ii)]{N-sur}.
This implies the following.
 
\begin{lemma} \label{l:al}
Suppose that {\aem} admits approximation ratio~$\al$.
If {\kcs} admits approximation ratio $\be$, then {\akcs} admits approximation ratio 
$O\left(\log k +\log\min\{\th,n\}\right) +\be \cdot (\th+1)$,
and a similar statement holds for {\akos}.
\end{lemma}

Let us briefly review the approximability status of {\kos} and {\kcs}.
The directed version of {\kos} admits a polynomial time algorithm \cite{FT}, 
and this implies approximation ratio $2$ for the undirected version.
One the other hand, the current approximability status of {\sc $k$-Connected Subgraph} is somewhat more complicated;
the following bounds on $\be$ are known (see a survey in \cite{N-kcs}).
\begin{itemize}
\item
$\be=\lceil \f{k+1}{2}\rceil$ for $2 \leq k \leq 7$ \cite{KhR,ADNP,DN,KN-kcs}.
\item 
$\be=2(2+1/q)$, where $q \approx \f{1}{2}(\log_k n-1)$ is the largest integer such that $2^{q-1}k^{2q+1} \leq n$  \cite{N-kcs4,CV}.
In particular $\be=4+\epsilon$ for any constant $k$ and $\epsilon>0$. 
\item
$\be=O\left(\log k \log \f{n}{n-k}\right)$ for any $k$ \cite{N-comb}.
\end{itemize}
 
 Thus from Lemma~\ref{l:al} we get the following. 

\begin{corollary}
For $\th=O(\log n)$, 
{\akos} admits approximation ratio $O(\log n)$, and
{\akcs} admits approximation ratio $O(\log n)$ unless $k=n-o(n)$. 
\end{corollary}

\section{Proof of Theorem~\ref{t}} \label{s:t}

In this section we prove Theorem~\ref{t}. 
The following lemma shows that it is sufficient to prove just approximation ratio $O(\log k+\log \th)$,
since it implies approximation ratio $O(\log n)$. 

\begin{lemma} \label{l:red}
If {\aec} admits approximation ratio $O(\log (k\th))$ then it also admits approximation ratio $O(\log n)$.
\end{lemma}
\begin{proof}
Let $(G=(V ,E),c,r)$ be an {\aem} instance.
Fix some optimal solution $J^*$. 
Let $M^*=\max\{\max\{c_e^u,c_e^v\}:e=uv \in J^*\}$ be the maximum $c$-cost of an end of an edge in $J^*$.
While $M^*$ is not known, it is sufficient to have some estimate $M$ for $M^*$, say $M^* \leq M \leq 2M^*$;
for that, we apply the procedure below for  every $M \in \{2^i: i=0, \ldots, \lceil \log C \rceil\}$, 
where $C=\max_{uv \in E}c_e^u$ is the maximum $c$-cost of an end of an edge in $E$, 
and return the best outcome.
So in what follows we will assume that we know an estimate $M$ for $M^*$ such that $M^* \leq M \leq 2M^*$.

Let $\rho$ be a parameter (eventually set to $\rho=O(\log n)$), let $\eps>0$ be another parameter, and let 
\[
\al=\f{\eps M}{n(\rho+1)} \ .
\]
Remove from $G$ all edges that have an end of cost greater than $M$.
Note that $J^*$ is a feasible solution of the obtained instance, since $M^* \leq M$. 
Define costs $\h{c}$ by
\[
\h{c}_{uv}^u=\max\{\lfloor c_{uv}^u/\al \rfloor,1\} \ .
\]
Let us denote by $\h{\ell}_J(v)=\max \{\h{c}_e^v: e \in \delta_J(v)\}$ the optimal 
assignment w.r.t. costs $\h{c}$ that activates a given edge set $J$. 

Let $J \subs E$ and let
\[
V_0=\{v \in V:\ell_J(v)/\al < 1\} \ \ \ \ \ V_1=\{v \in V:\ell_J(v)/\al \geq 1\} \ .
\]
Note that $\h{\ell}_J(V_0)=|V_0| \leq n$ and that $\h{\ell}_J(V_1) \leq \ell_J(V_1)/\al \leq \ell_J(V)/\al$. 
This implies 
\[
\h{\ell}_J(V) = \h{\ell}_J(V_0)+\h{\ell}_J(V_1) \leq n+\ell_J(V)/\al \ .
\]
Also note that $\ell_J(V_0)/\al \leq |V_0|$ and that $\ell_J(V_1)/\al \leq \h{\ell}_J(V_1)+|V_1|$. This implies 
\[
\ell_J(V)/\al=\ell_J(V_0)/\al+\ell_J(V_1)/\al \leq |V_0|+\h{\ell}_J(V_1) +|V_1| \leq n+\h{\ell}_J(V) \ .
\]
Summarizing, we have
\[
\ell_J(V)/\al-n \leq \h{\ell}_J(V) \leq \ell_J(V)/\al+n \ .
\]
From this we get that if $J$ is a $\rho$-approximate solution w.r.t. costs $\h{c}$ then
\[
\ell_J(V) \leq \al (\h{\ell}_J(V)+n) \leq  \al (\rho \h{\ell}_{J^*}(V)+n) \leq \al(\rho( \ell_{J^*}(V)/\al + n)+n)=\rho \ell_{J^*}(V)+\al n(\rho+1)
\]
Note that $M \leq 2\ell_{J^*}(V)$ and thus by the definition of $\al$ we have 
$\al n(\rho+1)=\eps M \leq 2 \eps \cdot \ell_{J^*}(V)$.
Consequently, 
\[
\ell_J(V) \leq \ell_{J^*}(V)(\rho+2\eps) \ .
\]

Finally, note that the maximum $\h{c}$-cost is bounded by 
$M/\al =n(\rho+1)/\eps$ while the minimum $\h{c}$-cost is at least $1$.
Thus the slope $\h{\th}$ of the obtained instance is bounded by
\[
\h{\th} \leq n(\rho+1) \ .
\]
Summarizing, the obtained instance has slope at most $n(\rho+1)$ and 
approximation ratio $\rho$ for the obtained instance implies 
approximation ratio $\rho+2\eps$ for the original instance.

Now let $\rho=O(\log n)$ and suppose that the obtained instance admits approximation ratio $O(\log (k \th))$. 
Then  the obtained instance admits also approximation ratio
\[
O(\log (k \h{\th}))=O(\log (kn (\rho+1))=O(\log (kn \log n))=O(\log n) \ .
\]
This implies approximation ratio $O(\log n)+2\eps=O(\log n)$ 
for the original instance, concluding the proof of the lemma. 
\end{proof}

\noindent
{\bf Remark.} The reduction in Lemma~\ref{l:red} applies to any activation network design problem.
The proof shows that for any$\eps>0$, if problem instances with $\th \leq n(\rho+1)/\eps$
admit approximation ratio $\rho$, then general instances admit approximation ratio $\rho+2\eps$. 

\medskip

By Lemma~\ref{l:red} it sufficient to show approximation ratio $O(\log (k\th))$. 
Following \cite{KMNT,CN}, we first reduce the problem with a loss of a factor of 
$2$ in the approximation ratio to the following particular case. 

\begin{center} \fbox{\begin{minipage}{0.98\textwidth} \noindent
\underline{\baem} \\
{\em Input:} \ A bipartite graph $G=(A \cup B,E)$ with activation costs $\{c_e^a,c_e^b\}$ for each edge 
\hphantom{Input:: }  $e=ab \in E$, and degree requirements $r=\{r_b:b \in B\}$ on $B$ only. \\
{\em Output:}   An $r$-edge-cover  $J \subs E$ of minimum activation cost.
\end{minipage}}\end{center}

The reduction is as follows. Add a copy $V'$ of $V$ and denote by $v'$ the copy of $v \in V$. 
Then replace every edge $uv$ by two edges $u'v$ and $v'u$ with activation costs as follows:
\begin{itemize} 
\item
$u'v$ has activation costs $\{c_{u'v}^{u'}=c_{uv}^u, c_{u'v}^v=c_{uv}^v\}$;
\item
$v'u$ has activation costs $\{c_{v'u}^{v'}=c_{uv}^v,c_{v'u}^u=c_{uv}^u\}$.
\end{itemize}
The degree requirements are $\{r_v:v \in V\}$, while nodes in $V'$ have no requirements.
It is not hard to see that ratio $\rho$ for the obtained {\baem} instance 
implies ratio $2\rho$ for the original instance, see \cite{KMNT,CN}. 

\medskip

So from now and on we will consider the {\baem} problem.
Whenever we consider an edge $e=ab$, it is assumed that $a \in A$ and $b \in B$. 
For $b \in B$, let $w_b$ be the $r_b$-th least cost at $b$ of an edge in $E$ incident to $b$,
where $w_b=0$ if $r_b=0$. 
The {\bf residual requirement} of $b$ w.r.t. an edge subset $J \subs E$ is defined by
\[
r^J_b= \max \{r_b-\deg_J(b), 0\} \ .
\] 
Define the following potential function on edge subsets $J \subs E$
\[ 
\Phi(J)=\sum_{b \in B} w_b  r^J_b
\]

Let ${\sf opt}$ denote the optimal solution value and let $\tau$ be an estimation for ${\sf opt}$.
The main step of the algorithm is given in the following lemma, which we will prove later. 

\begin{lemma} \label{l:1}
For any $\eps>0$ there exists a polynomial time algorithm that given 
an edge set $J \subs E$,
a parameter $\gamma>1$, 
and an integer $\tau$, 
returns an edge set $I \subs E \sem J$ such that the following holds:
\begin{enumerate}[(i)]
\item
$\ell_I(B) \leq \ga \tau$.
\item
$\ell_I(A) \leq \tau$ if $\tau \geq {\sf opt}$.
\item
$\Phi(J \cup I) \leq \al \cdot \Phi(J)$ if $\tau \geq {\sf opt}$,  
where $\al=1-\left(1-\frac{1}{\gamma}\right)\left(1-\frac{1}{e}-\eps\right)$. 
\end{enumerate}
\end{lemma}

The next lemma bounds the activation cost of a feasible solution obtained by picking 
edges with least cost at $B$-nodes. 

\begin{lemma} \label{l:2} 
Let $J \subs E$ and let $F$ be an edge set obtained by picking for every $b \in B$
a set of $r^J_b$ edges incident to $b$ in $E \sem J$ of minimal cost at $b$. 
Then $J \cup F$ is an $r$-edge-cover and: 
\begin{enumerate}[(i)]
\item
$\ell_F(B) \leq {\sf opt}$.
\item
$\ell_F(A) \leq \th \cdot \Phi(J)$.
\end{enumerate}
\end{lemma}
\begin{proof}
Since $F$ is an $r^J$-edge cover, $J \cup F$ is an $r$-edge-cover.
By the definition of $F$, $\ell_F(b) \leq w_b$ for all $b \in B$. 
Any $r$-edge-cover has activation cost at least $\sum_{b \in B}w_b$. 
Thus we have 
\[
\ell_F(B)=
\sum_{b \in B} \ell_F(b) \leq 
\sum_{b \in B}w_b  \leq 
{\sf opt} \ .
\]

We prove part (ii). Note that:
\begin{itemize}
\item
$c_{ab}^a \leq \th c_{ab}^b$ for every $ab \in F$. 
\item
$\displaystyle \sum_{ab \in F} c_{ab}^b \leq w_b r^J_b$ for every $b \in B$.
\end{itemize}
From this we get
\[
\ell_F(A)=
\sum_{a \in A} \ell_F(a) \leq 
\sum_{ab \in F} c_{ab}^a \leq 
\sum_{ab \in F} \th c_{ab}^b =
\th \sum_{b \in B} \sum_{ab \in F} c_{ab}^b \leq 
\th \sum_{b \in B} w_b r^J_b =
\th \cdot \Phi(J) \ ,
\]
as required.
\end{proof}

Theorem \ref{t} is deduced from Lemmas \ref{l:1} and \ref{l:2} as follows.
We let $\gamma$ to be a constant strictly greater than $1$, say $\gamma=2$, and we let $\eps=1/2-2/e$.
Then $\al=3/4$. Using binary search, we find the least integer $\tau$ such that the following procedure 
computes an edge set $J$ satisfying $\Phi(J) \leq \tau/\th$.

\medskip \medskip

\begin{algorithm}[H]
\caption{} \label{alg:main}
{\bf initialization}: $J \gets \empt$                                                                                    \\
{\bf loop:}         repeat $\lceil \log_{1/\al} (k\th) \rceil$ times                                           \newline
\hphantom{ii} apply the algorithm from Lemma~\ref{l:1}                                               \newline                         
\hphantom{ii} - If $\Phi(J \cup I)> \al \cdot \Phi(J)$ 
then return ``ERROR'' and stop  \newline
\hphantom{ii} - else do $J \gets J \cup I$.  
\end{algorithm}

\medskip \medskip

After computing $J$, we compute an edge $F$ set as in Lemma~\ref{l:2},
so $J \cup F$ is a feasible solution. 


\begin{lemma} \label{l:alg1}
If $\tau \geq {\sf opt}$ then Algorithm~\ref{alg:main} returns an edge set $J$ such that $\Phi(J) \leq \tau/\th$.
\end{lemma}
\begin{proof}
Note that 
\[
\Phi(\empt) = \sum_{b \in B} w_b r_b \leq  k \sum_{b \in B} w_b \leq  k \cdot {\sf opt} \ .
\]
From this we get that
\[
\di \Phi(J) \leq 
\Phi(\empt) \cdot \al^{\lceil \log_{1/\al} (k\th) \rceil} \leq 
\Phi(\empt) \cdot 1/k \th \leq 
k \cdot {\sf opt}/k\th={\sf opt}/\th \leq \tau/\th \ ,
\]
as required. 
\end{proof}

By Lemma~\ref{l:alg1}, the least integer $\tau$ for which the procedure does not return ``ERROR'' 
satisfies $\tau \leq {\sf opt}$. 
By Lemma~\ref{l:1}(i)  and since the number of iterations in Algorithm~\ref{alg:main} is $\lceil \log_{1/\al} (k \th) \rceil$, we have:
\[
\ell_J(A \cup B) \leq 
\lceil \log_{1/\al} (k \th) \rceil (1+\ga) \tau =
O(\log (k\th)) \cdot {\sf opt} \ .
\]
Also, by Lemmas \ref{l:2} and \ref{l:alg1} we have:
\[
\ell_F(A \cup B) =
\ell_F(A)+\ell_F(B) \leq 
\th \cdot \Phi(J)+{\sf opt} \leq 
\th \cdot \tau/\th+{\sf opt} \leq 
\tau+{\sf opt} \leq 
2 \cdot {\sf opt} \ .
\]
Consequently
\[
\ell_{J \cup F}(A \cup B) \leq 
\ell_J(A \cup B)+\ell_F(A \cup B) =
O(\log(k \th)) \cdot {\sf opt} +2{\sf opt} =
O(\log(k \th)) \cdot {\sf opt} \ .
\]

This concludes the proof of Theorem~\ref{t}, except that we need to prove Lemma~\ref{l:1}, which we will do in the next section. 

\section{Proof of Lemma~\ref{l:1}} \label{s:l}

It is sufficient to prove Lemma~\ref{l:1} for the residual instance with 
$E \gets E \sem J$ and $r \gets r^J$; namely, we may assume that $J = \empt$.
Let us use the notation 
$$
\Phi_0=\Phi({\empt})=\sum_{b \in B}w_b r_b \ .
$$
Then Lemma~\ref{l:1} can be restated as follows: \\
{\em For any $\eps>0$ there exists a polynomial time algorithm that given 
a parameter $\gamma>1$, 
and an integer $\tau$, 
returns an edge set $I \subs E$ such that the following holds:
\begin{enumerate}[(i)]
\item
$\ell_I(B) \leq \ga \tau$.
\item
$\ell_I(A) \leq \tau$ if $\tau \geq {\sf opt}$.
\item
$\Phi(I) \leq \al \cdot \Phi_0$ if $\tau \geq {\sf opt}$,  
where $\al=1-\left(1-\frac{1}{\gamma}\right)\left(1-\frac{1}{e}-\eps\right)$. 
\end{enumerate}
}

\begin{definition}
We say that an edge $e \in E$ incident to a node $b \in B$ is {\bf $\tau$-cheap} if 
$c_e^b \leq \ga \tau \cdot \cdot w_b r_b/\Phi_0$. Let $C$ denote the set of $\tau$-cheap edges in $E$,
namely
$$
C=\bigcup\limits_{b \in B} \left\{e \in \delta_E(b):c_e^b \leq \ga\tau \cdot \f{w_b r_b}{\Phi_0}\right\} \ .
$$ 
\end{definition}

By the definition of $\tau$-cheap edges and $\Phi_0$ we have
\[
\ell_C(B)=\sum_{b \in B}\ell_C(b) \leq 
\ga \tau \cdot \f{1}{\Phi_0} \sum_{b \in B} w_b r_b=
\ga \tau \f{1}{\Phi_0} \cdot \Phi_0 = \ga \tau \ .
\]
This implies that if we choose $I$ to be a subset of $\tau$-cheap edges, then the first condition $\ell_I(B) \leq \ga \tau$ in (i) will hold. 
The next lemma shows that if $\tau$ is not too small, then the  $\tau$-cheap edges in any feasible solution $F$  
reduce the potential by a factor of at least $1/\ga$. 

\begin{lemma} \label{l:cheap}
Let $F$ be an $r$-edge-cover and let $\tau \geq \ell_F(B)$.
Then $\Phi(C \cap F) \leq \Phi_0/\gamma$. 
\end{lemma}
\begin{proof}
Let $D=\{b \in B: \de_{F \sem C}(b) \neq \empt\}$. 
Since for every $b \in D$ there is an edge $e \in F \sem C$ incident to $b$ 
with $c^b_e > \f{\tau\ga}{\Phi_0} \cdot w_b r_b$,
we have $\ell_{F \sem C}(b) \geq \f{\tau \ga}{\Phi_0} \cdot w_b r_b$ for every $b \in D$.
Thus
\[
\tau \geq \ell_F(B) \geq 
\ell_{F \sem C}(B) = 
\sum_{b \in D} \ell_{F \sem C}(b) \geq 
\ga \tau \cdot \f{1}{\Phi_0} \sum_{b \in D} w_b r_b  \ .
\]
This implies that $\sum_{b \in D} w_b r_b \leq \Phi_0/\ga$. 
Note that for every $b \in B \sem D$,
$\de_F(b) \subs \de_C(b)$ and hence $r^{C \cap F}(b)=r^F(b)=0$.
Thus we get
\[
\Phi(C \cap F) = 
\sum_{b \in B} w_b r^{C \cap F}_b= 
\sum_{b \in D} w_b r^{C \cap F}_b \leq 
\sum_{b \in D} w_b r_b \leq \Phi_0/\ga \ .
\] 
This concludes the proof of the lemma.
\end{proof}

Let now $F^*$ be an optimal $r$-edge-cover and let $I^*=C \cap F^*$. 
Since $\ell_C(B) \leq \ga \tau$ (this is so for any $\tau$) and by Lemma~\ref{l:cheap} we have:
\begin{enumerate}[(i')]
\item
$\ell_{I^*}(B) \leq \ga \tau$.
\item
$\ell_{I^*}(A) \leq \tau$ if $\tau \geq {\sf opt}$.
\item
$\Phi(I^*) \leq \Phi_0/\ga$ if $\tau \geq {\sf opt}$. 
\end{enumerate}
This shows that there exists a ``good'' set of $\tau$-cheap edges, 
that satisfies property (iii) with a constant $1/\ga$ that is smaller than 
$\al=1-\left(1-\frac{1}{\gamma}\right)\left(1-\frac{1}{e}-\eps\right)$ in (iii). 
Unfortunately, we are not able to find such $I^*$ in polynomial time. 
However, we can find an approximate solution $I$, that by ``budget'' $\tau$ at $A$ still reduces the 
potential by a constant factor.  
The problem we need to solve is as follows. 

\medskip

\begin{center} \fbox{\begin{minipage}{0.960\textwidth}\noindent
\underline{\sc Bipartite Activation-Budgeted Maximum Edge-Multi-Coverage} \\
{\em Input:} \ \ A bipartite graph $G=(A \cup B,E)$ with edge-costs $\{c_e:e \in E\}$ and node-weights \hphantom{\em Input:} \ \
$\{w_v:v \in B\}$,  degree bounds $\{r(v):v \in B\}$, and a  budget $\tau$.\\
{\em Output:} Find $I \subseteq E$ with $\ell_I(A) \leq \tau$ that maximizes $\displaystyle \sum_{b \in B} w_b \cdot \min\{\deg_I(b),r_b\}$.
\end{minipage}}\end{center}

\begin{lemma}
{\sc Bipartite Activation-Budgeted Maximum Edge-Multi-Coverage} 
admits a $(1-1/e-\eps)$-approximation algorithm.
\end{lemma}
\begin{proof}
We will show that the problem can be cast as the one of maximizing a submodular function subject to 
one matroid constraint and one knapsack constraint, 
that in the value oracle admits a $(1-1/e-\eps)$-approximation algorithm \cite{CVZ}.

Let $\AA$ be the set of stars with center in $A$. 
For $S \in \AA$ with center $a$ let $\ell(S)=\max_{ab \in S}c_{ab}^a$ be the activation cost at $a$ of $S$. 
For $\SS \subs \AA$ and $b \in B$ let $\deg_\SS(b)$ denote the degree of $b$ in the union of the stars in $\SS$. 
Let
\[
f(\SS) = \sum_{b \in B} w_b \cdot \min\{\deg_{\SS}(b),r_b\}
\]
Consider some inclusion minimal solution $I$ to the problem.
Then $I$ can be partitioned into a collection $\SS(I)$ of stars with centers in $A$,
where no two stars have a node in $A$ in common. 
Thus the problem can be cast as maximizing $f(\SS)$ subject to two constraints:
\begin{enumerate}[(i)]
\item
No two stars have a common center.
\item
$\sum_{S \in \SS} \ell(S) \leq \tau$. 
\end{enumerate}

Note that the function $f(\SS)$ is submodular; this is so since for every $b \in B$,
$\deg_{\SS}(b)$ is submodular and thus $\min\{\deg_{\SS}(b),r_b\}$ is submodular.
Since non-negative linear combination of submodualar functions is also submodular,
we get that $f(\SS)$ is submodular. 
 
The constraints in (i) are matroid constraints, of the partition matroid,
where for every $a \in A$ we have a part of the stars with center $a$. 
The constraints in (ii) are knapsack constraints.
This concludes the proof of the lemma. 
\end{proof}

The following algorithm computes an edge set as in Lemma~\ref{l:1}. 
\begin{enumerate}
\item
Among the $\tau$-cheap edges, compute a $\left(1-\frac{1}{e}-\eps\right)$-approximate
solution $I$ to {\sc Bipartite Activation-Budgeted Maximum Edge-Multi-Coverage}.
\item
If $\Phi(I) \leq \al \Phi_0$ then return $I$, where 
$\al =1-\left(1-\frac{1}{\gamma}\right)\left(1-\frac{1}{e}-\eps\right)$; \\
Else declare ``$\tau<{\sf opt}$''.
\end{enumerate}

\vspace{0.2cm}

We have  $\ell_I(A) \leq \tau$ and  $\ell_I(B) \leq \gamma \tau$.
Now we show that if $\tau \geq {\sf opt}$ then $\Phi(I) \leq \al \Phi_0 $.
Let $F^*$ be some optimal solution. 
Then $\ell_{C \cap F^*}(A) \leq {\sf opt} \leq \tau$. 
By Lemma~\ref{l:cheap} $\Phi(C \cap F^*) \leq \Phi_0/\gamma$, namely, 
$C \cap F^*$ reduces $\Phi$ by at least $\Phi_0\left(1-\frac{1}{\gamma}\right)$.
Hence the $\left(1-\frac{1}{e} -\eps \right)$-approximate solution $I$ to 
{\sc Bipartite Activation-Budgeted Maximum Edge-Multi-Coverage}
reduces $\Phi_0$ by at least $\Phi_0\left(1-\frac{1}{e}-\eps\right)\left(1-\frac{1}{\gamma}\right)$.
Consequently, 
\[
\Phi(I) \leq \Phi_0-\Phi_0\left(1-\f{1}{e}-\eps\right)\left(1-\frac{1}{\gamma}\right)=\al \Phi_0 \ ,
\] 
as claimed.

\medskip

This concludes the proof of Lemma~\ref{l:1}, and thus also the proof of Theorem~\ref{t} is complete. 


\end{document}